\documentclass[runningheads]{llncs}
\usepackage[T1]{fontenc}
\usepackage{graphicx}
\usepackage{indentfirst}
\usepackage{misccorr}
\usepackage{amsmath}
\usepackage{amssymb}
\usepackage{float}
\usepackage{bm}
\usepackage{comment}
\usepackage[normalem]{ulem}
\usepackage{pgfplots}
\usepackage{tikz}
\usepackage{enumitem}
\usepackage[unicode, pdftex]{hyperref}
\usepackage{subcaption}
\usepackage{multirow}
\usepackage{amsmath}
\usepackage{subcaption}
\usepackage{algorithm}
\usepackage[noend]{algpseudocode}

\def\Return{\State\textbf{return}\ }

\usepackage{color}

\newcommand{\sep} {\mathrm{sep}\,}
\newcommand{\init} {\mathrm{init}\,}
\newcommand{\ord} {\mathrm{ord}\,}

\newcommand{\coeff}{\,\mathrm{coeff}\,}

\DeclareMathOperator{\im}{im}
\DeclareMathOperator{\bx}{\mathbf{x}}
\DeclareMathOperator{\by}{\mathbf{y}}
\DeclareMathOperator{\bg}{\mathbf{g}}
\DeclareMathOperator{\bz}{\mathbf{z}}
\DeclareMathOperator{\bu}{\mathbf{u}}
\DeclareMathOperator{\balpha}{\bm{\alpha}}
\DeclareMathOperator{\bbeta}{\bm{\beta}}

\DeclareMathOperator{\bmu}{\bm{\mu}}

\DeclareMathOperator{\KK}{\mathbb{K}}

\newcommand{\ideal}[1]{\ensuremath{( #1 )}}
\newcommand{\tuple}[1]{\ensuremath{[#1]^{T}}}

\DeclareMathOperator{\initial}{in}

\usepackage[textsize=tiny]{todonotes}

\newtheorem{notation}{Notation}

\begin{document}
\title{Support bound for differential elimination in polynomial dynamical systems}
%
%
\author{Yulia Mukhina \and
Gleb Pogudin}
\authorrunning{Y. Mukhina and G. Pogudin}
%
\institute{LIX, CNRS, \'Ecole polytechnique, Institute Polytechnique de Paris, Paris, France
\email{\{yulia.mukhina,gleb.pogudin\}@polytechnique.edu}}
\maketitle              
\begin{abstract}
We study an important special case of the differential elimination problem: given a polynomial parametric dynamical system $\bx' = \bg(\bmu, \bx)$ and a polynomial observation function $y = f(\bmu,\bx)$, find the minimal differential equation satisfied by $y$.
In our previous work \cite{mukhina2025projecting}, for the case $y = x_1$, we established a bound on the support of such a differential equation for the non-parametric case and showed that it can be turned into an algorithm via the evaluation-interpolation approach.
The main contribution of the present paper is a generalization of the aforementioned result in two directions: to allow any polynomial function $y = f(\bx)$, not just a single coordinate, and to allow $\bg$ and $f$ to depend on unknown symbolic parameters.
We conduct computation experiments to evaluate the accuracy of our new bound and show that the approach allows to perform elimination for some cases out of reach for the state of the art software.

\end{abstract}
\section{Introduction}

Elimination is a classical operation in computer algebra extensively studied for different classes of systems of equations.
For a given system $\mathbf{F}(\bx, \by) = 0$ in two groups of unknowns $\mathbf{x} = \tuple{x_1, \ldots, x_s}$ and $\mathbf{y} = \tuple{y_1, \ldots, y_\ell}$, the goal of elimination is to describe nontrivial equations of the form $\mathbf{g}(\by) = 0$ implied by the original system.
Prominent examples include Gaussian elimination for linear equations or resultants and Gr\"obner bases for polynomial systems.

The elimination problem for systems of algebraic differential equations goes back at least to the works of J. Ritt in the 1930s \cite{Ritt}.
Since then, it has been studied from the point of view of algorithmic development \cite{WANG2002,Boulier1,Boulier2,Hubert2000,Hubert2003b,diff_Thomas,Robertz2014,Rueda2016,simmons2023differentialeliminationalgebraicinvariants}, software implementation \cite{WANG2002,TDDS,blad}, and complexity analysis \cite{Gustavson2018,Grigorev1989,Li2015}, and it has been used in different application domains \cite{Boulier2007,Pascadi2021,Gerdt2019}.
Recently, particular attention has been drawn \cite{hong2020global,Dong2023,AitElManssour2025,Meshkat2012} to elimination for systems in \emph{the state-space form} which naturally arise in many applications in modeling and control:
\begin{equation}\label{eq:state_space}
\bx' = \mathbf{g}(\bx, \bu), \quad \by = \mathbf{f}(\bx, \bu),
\end{equation}
where $\bx$, $\by$, and $\bu$ are sets of differential variables describing the internal state, the observed output, and the external input, respectively.
In this setup, one is typically interested in eliminating the state variables $\bx$ because there is no experimental data available for them directly.

Until recently, both general-purpose elimination algorithms and those tailored specifically to the systems \eqref{eq:state_space} are based on arithmetic manipulations with the input differential equations and, thus, suffer from the \emph{intermediate expression swell}.
In our recent paper \cite{mukhina2025projecting}, we have proposed the following workaround building upon the observation that, for systems of the form \eqref{eq:state_space}, truncated power series solutions can be easily computed for any given initial conditions.
We have established a bound on the support of the resulting elimination polynomial \cite[Theorem 1]{mukhina2025projecting}, made an ansatz, and used power series solutions to find the unknown coefficients by solving a system of linear equations.
We have shown that the bounds we obtain are often sharp \cite[Theorems 2 and 3]{mukhina2025projecting}, and the resulting straightforward algorithm was able to perform elimination for systems out of reach for prior methods \cite[Section 9]{mukhina2025projecting}.
While the results obtained in \cite{mukhina2025projecting} provide a strong indication of the viability of the approach, it was restricted to a subclass of systems \eqref{eq:state_space} without $\mathbf{u}$, with $\mathbf{g}$ being polynomials, and with $\mathbf{f}$ being a scalar equal to one of the coordinates of $\bx$.

The main theoretical contribution of the present paper, Theorem \ref{th::bound-param}, is an extension of the bound from \cite{mukhina2025projecting} in two directions:
\begin{itemize}
    \item $\mathbf{f}$ is still a scalar but now can be an arbitrary polynomial in $\bx$;
    \item both $\mathbf{g}$ and $\mathbf{f}$ can depend on constant parameters $\bmu$, and the support is considered with respect to the parameters as well. 
\end{itemize}
The former extension not only allows nontrivial observation functions arising in modeling but also makes the bound applicable to effective computations with D-algebraic functions \cite{AitElManssour2025,kauers2025boundsdalgebraicclosureproperties,TeguiaTabuguia2025} (see Example \ref{ex:dalg_functions}).
The latter extension is the first step towards using the new elimination approach in the context of the structural parameter identifiability which was one of important recent applications of differential elimination \cite{Dong2023,Meshkat2012}.
The proof of the new bound builds upon and refines the techniques developed in \cite{mukhina2025projecting}, in particular, we employ multihomogeneous B\'ezout bound instead of the standard one. 

We show experimentally that the obtained bound is sharp in the non-parametric case and quite accurate in the other cases.
Furthermore, we use the updated bound to extend the implementation from \cite{mukhina2025projecting} accordingly and evaluate its performance. 
The new version of software is available at
\begin{center}
    \url{https://github.com/ymukhina/Loveandsupport/tree/y-input}
\end{center}

The rest of the paper is organized as follows.
Section \ref{sec:preliminaties} contains preliminaries on differential algebra and introduces formally the class of the systems considered in this paper.
In Section \ref{sec:main} we formulate the main theoretical result of the paper, a bound on the support of the result of elimination.
Section \ref{sec:experimental} summarizes the experimental study of the accuracy of the bound and poses some conjectures based on the given data.
The proof of the bound is given in Sections \ref{sec:proof1}-\ref{sec:proof2}.
We recall the elimination algorithm from \cite{mukhina2025projecting} and explain how it can be extended using the new bound in Section \ref{sec:algorithm}.
Section \ref{sec:implementation} describes the implementation of the extended algorithm and its performance.

\paragraph{Acknowledgements} The authors are grateful to Bertrand Teguia Tabuguia for useful discussion, in particular, related to Example \ref{ex:dalg_functions}.  The authors also thank Rafael Mohr for helpful conversations.
This work has been supported by the French ANR-22-CE48-0008 OCCAM and ANR-22-CE48-0016 NODE projects.



\section{Preliminaries}
\label{sec:preliminaties}
\begin{definition}[Differential rings and ideals]
\begin{itemize}
    \item A \emph{differential ring} $(R, \;')$ is a commutative ring with a derivation $': R \rightarrow R$, that is, a map such that, for all $a,b \in R$, $(a+b)' = a' +b'$ and $(ab)' = a'b + ab'$.  
    For $i > 0$, $a^{(i)}$ denotes the $i$-th order derivative of $a \in R$.
   \item A \emph{differential field} is a differential ring which is a field.
   \item Let $R$ be a differential ring. An ideal $I \subset R $ is called a \emph{differential ideal} if $a' \in I$ for every $a \in I$.
\end{itemize}
\end{definition}

For the rest of the paper, $\mathbb{K}$ will be a differential field of characteristic zero.

Let $x$ be an element of a differential ring. 
We denote $x^{(\infty)} := \{x,x',x'',x^{(3)},\ldots\}.$

\begin{definition}[Differential polynomials]
    Let $R$ be a differential ring. Consider a ring of
polynomials in infinitely many variables
\[
R[x^{(\infty)}] := R[x, x', x'', x^{(3)},\ldots]
\]
and extend \cite[§ 9, Prop. 4]{bourbaki1950elements} the derivation from $R$ to this ring by $(x^{(j)})':= x^{(j+1)}$. The resulting differential ring is called the \emph{ring of differential polynomials in $x$ over $R$}.
The ring of differential polynomials in several variables is defined by iterating this construction.
\end{definition}

\begin{notation}
    One can verify that $\ideal{f_1^{(\infty)},\ldots, f_s^{(\infty)}}$ is a differential ideal 
    for every $f_1, \ldots , f_s \in  R[x_1^{(\infty)},\ldots, x_n^{(\infty)}] $. Moreover, this is the minimal differential ideal containing $f_1, \ldots , f_s$,
    and we will denote it by $\ideal{f_1,\ldots, f_s}^{(\infty)}$.
\end{notation}

\begin{notation}[Saturation]
Let $I$ be an ideal in the ring $R$, and $a \in R$. Denote
$$
I : a^{\infty} := \{ b \in R\; | \; \exists N \colon\, a^Nb \in I\}.
$$
The set $I : a^{\infty}$  is also an ideal in $R$.
If $I$ is a differential ideal, then $I : a^{\infty}$ is also a differential ideal.
\end{notation}

 \begin{definition}
    Let $P \in \mathbb{K}[\bx^{(\infty)}]$ be a differential polynomial in $\bx = \tuple{x_1, \ldots, x_n}$.
    \begin{enumerate}
        \item For every $1 \leqslant i \leqslant n$, we will call the largest $j$ such that $x_i^{(j)}$ appears in $P$ the \emph{order} of $P$ 
        respect to $x_i$ and denote it by $\ord_{\!x_i} P$; if $P$ does not involve $x_i$, we set $\ord_{\!x_i} P := -1$.

        \item For every $1 \leqslant i \leqslant n$ such that $x_i$ appears in $P$, the \emph{initial} of $P$ with respect to $x_i$ is the
        leading coefficient of $P$ considered as a univariate polynomial in $x_i^{(\ord_{\!x_i} P)}$. 
        We denote it by $\init_{\!x_i} P$.
        
        \item The \emph{separant} of $P$ with respect to $x_i$ is
        $$
        \sep_{\!x_i} P := \dfrac{\partial P}{\partial x_{i}^{(\ord_{\!x_i} P)}} .
        $$
    \end{enumerate}
    \end{definition}

    \begin{notation}\label{not:Ifg}
          Consider an ODE system (in the \emph{state-space form}): 
          \begin{equation}\label{eq:ode_main}
          \bx' = \bg(\bmu, \bx), \quad y = f(\bmu, \bx),
          \end{equation}
          where $\bx = [x_1, \ldots, x_n]^T$, $\bmu = [\mu_1, \ldots, \mu_r]^T$, $g_1, \ldots, g_n, f \in \mathbb{K}[ \bmu, \bx]$.
          These differential equations can be viewed as differential polynomials in the differential polynomial ring $\mathbb{K}[\bmu, \bx^{(\infty)}, y^{(\infty)}]$ over a differential ring $\mathbb{K}[\bmu]$ (equipped with derivation via $\bmu' = 0$).
          The differential ideal $(\bx' - \bg, y - f)^{(\infty)} \in \mathbb{K}[\bmu, \bx^{(\infty)}, y^{(\infty)}]$ describing the solutions of this system will be denoted by $I_{\bg, f}$.
    \end{notation}

    By \cite{Moog1990} (see also \cite[Lemma 3.2]{hong2020global}), the ideal $I_{\bg, f}$ is prime.
    The elimination problem we study in this paper is, for a system \eqref{eq:ode_main}, to eliminate all the $\bx$-variables. 
    In other words, we want to describe a differential ideal
	\begin{equation} \label{id}
		I = I_{\bg, f} \cap \mathbb{K}[\bmu, y^{(\infty)}].
	\end{equation}
    Since $I_{\bg, f}$ is prime, the elimination ideal is prime as well.

    \begin{definition}[Minimal polynomial]\label{not:minpoly}
        The \emph{minimal polynomial} $f_{\min}$ of the prime ideal \eqref{id} is a polynomial in \eqref{id} of the minimal order and then the minimal total degree.
        It is unique up to a constant factor \cite[Proposition 1.27]{pogudin2023differential}.
    \end{definition}

    \begin{proposition}[{{\cite[Proposition 1.15]{pogudin2023differential}}}]\label{prop::idealstructure} 
        The prime ideal \eqref{id} is uniquely determined by its minimal polynomial $f_{\min}$. 
        More precisely:
        \begin{equation*} \label{diff_ideal}
            I = \ideal{f_{\min}}^{(\infty)} : (\sep_{\!y}(f_{\min}) \,\init_{\!y}(f_{\min}))^{\infty}.
        \end{equation*}
    \end{proposition}

\begin{example}
For a simple example of such representation consider the elimination problem for the following model with parameters $\mu_1$ and $\mu_2$:
\[ x'_1 = \mu_1 x_2, \quad x'_2 = \mu_2 x_1, \quad y = x_1 + x_2. \]
$f_{\min}$ can be obtained via double differentiation: 
\[ (y - (x_1 + x_2))'' = (y' - \mu_1 x_2 - \mu_2 x_1)' = y'' - \mu_1 \mu_2 (x_1 + x_2) = y'' - \mu_1 \mu_2 y.\] Thus, $f_{\min} = y'' - \mu_1 \mu_2 y$.
    
\end{example}

\begin{example}[Closure properties of D-algebraic functions]\label{ex:dalg_functions}
    Consider the tangent $\tan (t)$ and hyperbolic tangent $\tanh(t)$ functions. 
    They both are known to satisfy first-order differential equations $\tan'(t) = 1 + \tan^2(t)$ and $\tanh'(t) = 1 - \tanh^2(t)$ (in other words, they are \emph{D-algebraic}).
    We can now, following \cite{AitElManssour2025,TeguiaTabuguia2025}, find a differential equation satisfied, for example, by their product $\tan(t) \tanh(t)$ by considering the elimination problem for the following model:
    \[
    x_1' = 1 + x_1^2, \quad x_2' = 1 - x_2^2, \quad y = x_1 x_2.
    \]
    The triple $(x_1, x_2, y) = (\tan(t), \tanh(t), \tan(t) \tanh(t))$ is a solution of this system.
    A computation using any differential elimination algorithm (in particular, the one we describe in Section \ref{sec:algorithm}) shows that the minimal differential equation satisfied by the $y$-component of every solution (and, thus, vanishing at $\tan(t) \tanh(t)$) is:
    \[
    y^6 - y^4 y'' + y^4 + 2 y^3 (y')^2 + \frac{1}{4}(y^2 - 1) (y'')^2 - y^2 - y (y')^2 y'' - 2 y (y')^2 + (y')^4 + y'' - 1 = 0.
    \]
\end{example}


\section{Main Result}
\label{sec:main}

\begin{theorem} \label{th::bound-param}
  Let $f, g_1, \ldots, g_n$ be polynomials in $\mathbb{K}[\mu_1, \ldots, \mu_r, x_1, \ldots, x_n] = \mathbb{K}[\bmu, \bx]$.
   Denote
   \[
   d_\alpha := \deg_{\alpha} f\quad \text{ and }\quad D_\alpha := \max_{1 \leqslant i \leqslant n} \deg_{\alpha} g_i \quad\text{ for } \alpha = \bmu \text{ or } \alpha = \bx,
   \]
   and assume that $d_{\bx}, D_{\bx} > 0$.
   Let $f_{\min}$ be the minimal polynomial of $I_{\bg, f} \cap \mathbb{K}[\bmu, y^{(\infty)}]$ (see Notation \ref{not:Ifg} and \ref{not:minpoly}).
   
  Consider a positive integer $\nu$ such that $\ord f_{\min} \leqslant \nu$ ($\nu = n$ can always be used).
  Then for every monomial $\left(\prod\limits_{i=1}^{r} \mu_i^{\ell_i}\right) y^{e_0} (y')^{e_1} \ldots (y^{(\nu)})^{e_{\nu}}$ in $f_{\min}$
  the following inequalities hold:
  \begin{align} 
    \label{eq::bound_multi} & \sum\limits_{i=0}^{r} \ell_i + \sum\limits_{i=0}^{\nu}(d_{\bmu} + i D_{\bmu}) e_i \leqslant \sum\limits_{i=0}^{\nu}(d_{\bmu} + i D_{\bmu}) 
      \prod\limits_{j = 0, j \neq i}^{\nu}(d_{\bx} + j(D_{\bx}-1)), \\
    \label{eq::bound_bezout_global}& \sum\limits_{i=0}^{r} \ell_i + \sum\limits_{i=0}^{\nu}(d_{\bx} + i(D_{\bx}-1) ) e_i \leqslant
      \prod_{i=0}^{\nu}( d_{\bx} + d_{\bmu} + i(D_{\bx} + D_{\bmu} - 1)),\\
    \label{eq::bezout_y_only}&  \sum\limits_{i=0}^{\nu}(d_{\bx} + i(D_{\bx}-1) ) e_i \leqslant
      \prod\limits_{i = 0}^{\nu}(d_{\bx} + i(D_{\bx} - 1)).
  \end{align}

\end{theorem} 
We give the proof of the theorem in Section \ref{sec:proof2}.

\begin{corollary} \label{cor::bound}
  Let $f, g_1, \ldots, g_n$ be polynomials in $\mathbb{K}[x_1, \ldots, x_n] = \mathbb{K}[\bx]$
  such that $d := \deg f > 0$ and $D := \max_{1 \leqslant i \leqslant n} \deg g_i > 0$.
  Let $f_{\min}$ be the minimal polynomial of $I_{\bg, f} \cap \mathbb{K}[y^{(\infty)}]$ (see Notation \ref{not:Ifg} and \ref{not:minpoly}).
  Consider a positive integer $\nu$ such that $\ord f_{\min} \leqslant \nu$ ($\nu = n$ can always be used).
  Then for every monomial $y^{e_0} (y')^{e_1} \ldots (y^{(\nu)})^{e_{\nu}}$ in $f_{\min}$ the following inequality holds:
  \begin{equation}
    \sum\limits_{k=0}^{\nu}\bigl( d + k(D-1)\bigr)e_k \leqslant \prod\limits_{k=0}^{\nu}(d + k(D-1)).
  \end{equation}
  Note that this is essentially the inequality \eqref{eq::bezout_y_only} from Theorem \ref{th::bound-param}.
\end{corollary}

\begin{proof}(Corollary \ref{cor::bound})
Corollary \ref{cor::bound} is a special case of Theorem \ref{th::bound-param}, obtained by setting $r = 0$ and $d_{\bmu} = D_{\bmu} = 0$.
\end{proof}

\section{Experimental results}\label{sec:experimental}

In this section, we report the results of computational experiments aiming at evaluating the bounds given by Theorem \ref{th::bound-param} and Corollary \ref{cor::bound}.
For each of the experiments, we fixed some values (in the notation of Theorem \ref{th::bound-param}) of $n, |\bmu|, d_{\bx}, D_{\bx}, d_{\bmu}, D_{\bmu}$ and generated random dense parametric ODE models $\bx' = \bg(\bmu, \bx), y = f(\bmu, \bx)$ with the corresponding dimensions and degrees by sampling the coefficients uniformly at random from $[-1000, 1000] \cap \mathbb{Z}$.
For each such case, we report the following quantities:
\begin{itemize}
    \item \emph{\# terms in the bound}: the number of integer points inside the bound for the Newton polytope given by Theorem \ref{th::bound-param} (this be always a finite number);
    \item \emph{\# terms in the NP of $f_{\min}$}: the number of lattice points in the Newton polytope of the actual minimal polynomial (computed using the algorithm described in Section \ref{sec:algorithm});
    \item \emph{\# terms in $f_{\min}$}: the number of monomials in the actual minimal polynomial;
    \item \emph{\%}: the ratio between the number of monomials in $f_{\min}$ and the number of monomials in the bound from Theorem \ref{th::bound-param}. 
\end{itemize}

\begin{table}[H]
  \begin{subfigure}{0.57\textwidth}
    \begin{tabular}{ |c|c|c|c|c| } 
      \hline
    \multirow{2}{*}{$[D_{\bx}, d_{\bx}]$}  & \multicolumn{3}{|c|}{\# of terms}& \multirow{2}{*}{\%} \\ \cline{2-4}
      & Corollary \ref{cor::bound} & NP of $f_{\min}$ & $f_{\min}$ & \\
    \hline\hline	
    [1,1]& 4 & 4 & 4 & 100\%\\ 
    \hline
    [2,1]& 23 & 23 & 23 & 100\%\\ 
    \hline
    [2,2] & 169 & 169 & 169 & 100\%\\ 
    \hline
    [2,3] & 815 & 815 & 815 & 100\%\\ 
    \hline
    [2,4] & 2911 & 2911 & 2911 & 100\% \\ 
    \hline
    [2,5] & 8389 & 8389 & 8389 & 100\%\\ 
    \hline
    \hline
    [3,1] & 87 &  87 &  87  & 100\%\\ 
    \hline
    [3,2] & 575 &  575 & 575  &  100\% \\ 
      \hline
    [3,3] & 2287 & 2287 & 2287 & 100\% \\
      \hline
    [3,4] & 7153 & 7153 & 7153 & 100\% \\
      \hline
    [3,5] & 18325 & 18325 & 18325 & 100\% \\
      \hline
      \hline
    [4,1] & 241 & 241 & 241 & 100\% \\
    \hline
    [4,2] & 1417 & 1417 & 1417 & 100\% \\
      \hline
  \end{tabular}
  \caption{Bound for the dimension $n = 2$} \label{tab:non-par-1}
  \end{subfigure}
  \begin{subfigure}{0.57\textwidth}
    \centering
      \begin{tabular}{ |c|c|c|c|c| } 
        \hline
      \multirow{2}{*}{$[D_{\bx}, d_{\bx}]$}  & \multicolumn{3}{|c|}{\# of terms}& \multirow{2}{*}{\%} \\ \cline{2-4}
        & Corollary \ref{cor::bound} & NP of $f_{\min}$ & $f_{\min}$ & \\
      \hline\hline	
      [1,1] & 5 & 5  & 5  &  100\% \\ 
      \hline
      [1,2] & 495 & 495  & 495  &  100\% \\ 
      \hline
      [1,3] & 31465 &  31465 & 31465  &  100\% \\ 
      \hline
      \hline
      [2,1]& 1292 & 1292 & 1292 & 100\%\\ 
      \hline
      \hline
      [3,1] & 65637 & 65637  &  65637  & 100\%\\ 
      \hline
    \end{tabular}
    \caption{Bound for the dimension $n = 3$} \label{tab:non-par-2}
    \end{subfigure}
    \caption{Numerical values of the bound in the non-parametric ($|\bmu| = 0$) case} \label{tab:non-par}
    \end{table}

The numbers are consistent over several independent runs, so are equal to the generic values with high probability.
The considered cases can be classified as follows.
\begin{enumerate}
    \item \emph{Nonparametric systems} (i.e. $|\bmu| = 0$). 
    In this case, we always have $d_{\bmu} = D_{\bmu} = 0$.
    The results for different pairs $[D_{\bx}, d_{\bx}]$ are summarized in Tables \ref{tab:non-par-1} and \ref{tab:non-par-2} for $n = 2$ and $n = 3$, respectively.

    \item \emph{Parametric systems} with $n = 2$.
    Here we consider the case $D_{\bmu} = 1$, that is, when parameters appear only linearly in dynamics (as is often the case in practice) and split it into two subcases $d_{\bmu} = 0$ and $d_{\bmu} = 1$ depending on whether the parameters appear in the output or not.
    The results for these subcases (for $|\bmu| = 1$ and $|\bmu| = 2$) are given in Tables \ref{tab:par-1} and \ref{tab:par-2}, respectively.
\end{enumerate}

    \begin{table}[H]
  \begin{subfigure}{0.57\textwidth}
    \begin{tabular}{ |c|c|c|c|c| } 
      \hline
    \multirow{2}{*}{$[D_{\bx}, d_{\bx}]$}  & \multicolumn{3}{|c|}{\# of terms}& \multirow{2}{*}{\%} \\ \cline{2-4}
      & Theorem \ref{th::bound-param} & NP of $f_{\min}$ & $f_{\min}$ & \\
    \hline\hline
    [1,1]& 13 & 9 & 9 & 69\%\\ 
    \hline
    [1,2]& 350 & 280 & 280 & 80\%\\ 
    \hline
    [1,3] & 4675& 4015& 4015& 86\%\\ 
    \hline
    \hline
    [2,1] & 152 & 129 & 129 & 85\%\\ 
    \hline
    [2,2] & 2772 & 2434 & 2434 & 88\%\\ 
    \hline
    \hline
    [3,1] & 848 & 761 & 761 & 90\%\\ 
    \hline
    [3,2] & 12905 & 11755 & 11755 & 91\%\\ 
    \hline
    \hline
    [4,1] & 3088 & 2847 & 2847 & 92\%\\ 
    \hline
  \end{tabular}
  \caption{$|\bmu| = 1$} 
  \end{subfigure}
  \begin{subfigure}{0.57\textwidth}
    \centering
      \begin{tabular}{ |c|c|c|c|c| } 
        \hline
      \multirow{2}{*}{$[D_{\bx}, d_{\bx}]$}  & \multicolumn{3}{|c|}{\# of terms}& \multirow{2}{*}{\%} \\ \cline{2-4}
        & Theorem \ref{th::bound-param} & NP of $f_{\min}$ & $f_{\min}$ & \\
      \hline\hline
      [1,1]& 29 & 16 & 16 & 55\%\\ 
        \hline
      [1,2]& 2002 & 1337 & 1337 & 67\%\\ 
        \hline
      [1,3]& 53779 & 40414 & 40414 & 75\%\\ 
        \hline
        \hline
      [2,1]& 594 & 442 & 442 & 74\%\\ 
        \hline
      [2,2]& 24769 & 19394 & 19394 & 78\%\\ 
        \hline
        \hline
      [3,1] & 4665 & 3817 & 3817 & 82\%\\ 
        \hline 
        \hline 
      [4,1] & 21816 & 18728 & 18728 & 86\%\\ 
        \hline 
    \end{tabular}
    \caption{$|\bmu| = 2$} 
    \end{subfigure}
    \caption{Numerical values of the bound for the dimension $n = 2$, $d_{\bmu} = 0, D_{\bmu} = 1$} \label{tab:par-1}
    \end{table} 

    \begin{table}[H]
  \begin{subfigure}{0.57\textwidth}
    \begin{tabular}{ |c|c|c|c|c| } 
      \hline
    \multirow{2}{*}{$[D_{\bx}, d_{\bx}]$}  & \multicolumn{3}{|c|}{\# of terms}& \multirow{2}{*}{\%} \\ \cline{2-4}
      & Theorem \ref{th::bound-param} & NP of $f_{\min}$ & $f_{\min}$ & \\
    \hline\hline
    [1,1]& 22 & 10 & 20 & 45\%\\ 
    \hline
    [1,2]& 665 & 525 & 525 & 79\%\\ 
    \hline
    [1,3]& 9130 & 8030 & 8030 & 88\%\\ 
    \hline
    \hline
    [2,1]& 340 & 248 & 248 & 73\%\\ 
    \hline
    [2,2]& 6088 & 5243 & 5243 & 86\%\\ 
    \hline
    \hline
    [3,1]& 2318 & 1883 & 1883 & 81\%\\ 
    \hline
    [3,2]& 31825 & 28375 & 28375 &  89\%\\ 
    \hline
    \hline
    [4,1]& 9973 & 8527 & 8527 & 86\%\\ 
    \hline
  \end{tabular}
  \caption{$|\bmu| = 1$} 
  \end{subfigure}
  \begin{subfigure}{0.57\textwidth}
    \centering
      \begin{tabular}{ |c|c|c|c|c| } 
        \hline
      \multirow{2}{*}{$[D_{\bx}, d_{\bx}]$}  & \multicolumn{3}{|c|}{\# of terms}& \multirow{2}{*}{\%} \\ \cline{2-4}
        & Theorem \ref{th::bound-param} & NP of $f_{\min}$ & $f_{\min}$ & \\
      \hline\hline	
      [1,1]& 74 & 20 & 20 & 27\%\\
        \hline
      [1,2]& 6790 & 4340 & 4340 & 64\%\\
    \hline
    \hline
    [2,1]& 2717 & 1495 & 1495 & 55\%\\
    \hline
    \hline
    [3,1]& 32465 & 21745 & 21745 & 67\%\\ 
    \hline
    \end{tabular}
    \caption{$|\bmu| = 2$} 
    \end{subfigure}
    \caption{Numerical values of the bound for the dimension $n = 2$, $d_{\bmu} = D_{\bmu} = 1$} \label{tab:par-2}
    \end{table} 

The reported numerical data allows us to make the following observations.

\begin{itemize}
    \item The Newton polytope given by Corollary \ref{cor::bound} (that is, in the non-parametric case) \emph{coincides} with the actual Newton polytope.
    We believe that this can be proved using the methods developed to establish the tightness of more specialized bounds in our previous paper \cite{mukhina2025projecting}.

    \item In all the experiments, the minimal polynomial \emph{is dense with respect to its Newton polytope} (curiously, this is not always the case for the slightly different shape of the system in \cite[Table 1]{mukhina2025projecting}).
    We conjecture that this is always the case.
    This indicates that Newton polytope is an adequate tool to estimate the support of minimal polynomials.

    \item The accuracy of our bound (given in the \% column) \emph{increases towards $100\%$} when any of $d_{\bx}$ or $D_{\bx}$ increases.
    We conjecture that the bound accuracy reaches $100\%$ in any of the limits $d_{\bx} \to \infty$ or $D_{\bx} \to \infty$.
\end{itemize}


\section{Reduction to polynomial elimination}\label{sec:proof1} 
In this section we will explain how we reduce the differential elimination problem 
to a polynomial elimination problem using the approach similar to \cite{mukhina2025projecting}. 
We will fix some notation used throughout the rest of the paper.

For vectors $\balpha = [\alpha_1, \ldots, \alpha_r]^T \in \mathbb{Z}_{\geqslant 0}^r, \bbeta = [\beta_1, \ldots, \beta_n]^T \in \mathbb{Z}_{\geqslant 0}^n$ we denote
\[ 
    \bmu^{\balpha} \bx^{\bbeta} := \prod\limits_{i = 1}^r \mu_i^{\alpha_i} \prod_{j = 1}^n x_j^{\beta_j}, \quad |\balpha| := \sum_{i=1}^r \alpha_i \text{ and } |\bbeta| := \sum_{i=1}^n \beta_i.
\]

In the present paper, we consider an ODE system with parameters (see Notation \ref{not:Ifg})
\begin{equation} \label{eq::ODE-1}
  \bx' = \bg(\bmu, \bx) \text{ with } \bx = [x_1, \ldots, x_n]^T, \bmu = [\mu_1, \ldots, \mu_r]^T, g_1, \ldots, g_n \in \mathbb{K}[
  \bmu, \bx],
\end{equation}
and the problem of finding a differential equation for 
\begin{equation} \label{eq::ODE-2}
  y = f(\bmu, \bx),\; \; f \in \mathbb{K}[
  \bmu, \bx].
\end{equation}

\begin{notation} \label{not::LieOperator}
  Consider a polynomial vector field $ \bx' = \bg(\bmu, \bx)$ with  $\bx = [x_1, \ldots, x_n]^T$, $\bmu = [\mu_1, \ldots, \mu_r]^T$, $g_1, \ldots, g_n \in \mathbb{K}[\bmu, \bx]$.
  \begin{itemize}
     \item  For a polynomial $q := \sum\limits_{\balpha, \bbeta} k_{\balpha, \bbeta}  \bmu^{\balpha} \bx^{\bbeta} \in \mathbb{K}[\bmu, \bx]$ with $k_{\balpha, \bbeta} \in \mathbb{K}$, we denote $q^{\partial} := \sum\limits_{\balpha, \bbeta} k_{\balpha, \bbeta}'  \bmu^{\balpha} \bx^{\bbeta}$.
    \item We denote the Lie derivative operator $\mathcal{L}_{\bg} \colon \mathbb{K}[\bmu, \bx] \mapsto \mathbb{K}[\bmu, \bx]$
  by $\mathcal{L}_{\bg}(q) := \sum\limits_{i=1}^n g_i \dfrac{\partial q}{\partial x_i} + q^{\partial}.$
  \end{itemize}
\end{notation}

\begin{lemma}[{{cf. \cite[Lemma 1]{mukhina2025projecting}}}] \label{lem::ideal-equivalence}
  For the system \eqref{eq::ODE-1}-\eqref{eq::ODE-2} for every $s \geqslant 0$:
  \[ 
  (y - f, y' - \mathcal{L}_{\bg}(f), \ldots, y^{(s)} - \mathcal{L}^s_{\bg}(f)) = 
  I_{\bg, f} \cap \mathbb{K}[\bmu, \bx, y^{(\leqslant s)}].
  \]
\end{lemma}
\begin{proof}
  We denote 
  \[ J :=  (y - f, y' - \mathcal{L}_{\bg}(f), \ldots, y^{(s)} - \mathcal{L}^s_{\bg}(f)). \]
  and 
  \[ I := I_{\bg, f} \cap \mathbb{K}[\bmu, \bx, y^{(\leqslant s)}]. \]
  First, we prove that for every $0 \leqslant k \leqslant s$, we have $y^{(k)} - \mathcal{L}_{\bg}^{(k)} \in I$.
  We show this via induction on $k$. For the base case $k = 0$, we have $y - f \in I$. By the induction hypothesis
  for some $0 \leqslant k \leqslant s-1$ we have $p := y^{(k)} - \mathcal{L}^{k}_{\bg}(f) \in I$.
  We note that $p' \in I$ and $p' = y^{(k+1)} - \sum_{i=1}^n x_i' \frac{\partial}{\partial x_i} \mathcal{L}_{\bg}^k(f) - (\mathcal{L}_{\bg}^k(f))^{\partial}$. 
  Since $x'_i \equiv g_i \pmod I$, we get $p' = y^{(k+1)} - \mathcal{L}_{\bg}^{k+1}(f) \in I$.
  Hence, all generators of the ideal $J$ belong to $I$, so $J \subset I$.
  
  For the reverse inclusion, we proceed by contradiction. 
  Let $p$ be a polynomial in the ideal $I$ such that $p \notin J$.
  We fix the monomial ordering on $\mathbb{K}[\bmu, \bx, y^{(\leqslant s)}]$ to be the lexicographic monomial
  ordering with
  \[ 
  y^{(s)} > y^{(s-1)} > \ldots > y > x_1 > x_2 > \ldots > x_n > \mu_1 \ldots > \mu_r.
  \]
  The leading term of $y^{(i)} - \mathcal{L}_{\bg}^i(f)$ is $y^{(i)}$, so the leading terms of all generators of $J$ are distinct variables.
  Hence this set is a Gr\"obner basis of $J$ by the first Buchberger criterion \cite{buchberger1979criterion}. 
  The result of the reduction of $p$ with respect to 
  the Gr\"obner basis belongs to $\mathbb{K}[\bmu, \bx]$ and is distinct from zero. 
  Thus, we get a contradiction 
  with $p \in I$ because $I \cap \mathbb{K}[\bmu, \bx] = {0}$ by \cite[Lemmas 3.1 and 3.2]{hong2020global}.
\end{proof}

\begin{corollary}
 The minimal polynomial $f_{\min}$ in $I_{\bg, f} \cap \mathbb{K}[\bmu, y^{(\infty)}]$ with $s := \ord f_{\min}$ is the generator of the
   principal ideal $(y - f, y' - \mathcal{L}_{\bg}(f), \ldots, y^{(s)} - \mathcal{L}_{\bg}^s(f)) \cap \mathbb{K}[\bmu, y^{(\leqslant s)}]$.
\end{corollary}

\begin{proof}
    By Lemma \ref{lem::ideal-equivalence} we have
    \[I_{\bg, f} \cap \mathbb{K}[\bmu, y^{(\leqslant s)}] = (y - f, y' - \mathcal{L}_{\bg}(f), \ldots, y^{(s)} - \mathcal{L}^s_{\bg}(f)) \cap \mathbb{K}[\bmu, y^{(\leqslant s)}].\]
    We consider a polynomial map 
    \[ \varphi  : \mathbb{A}^{n} \rightarrow \mathbb{A}^{s+1}, \]
    such that
       \[ (x_1, \ldots, x_n) \mapsto (f, \mathcal{L}_{\bg}(f), \ldots, \mathcal{L}^{(s)}_{\bg}(f)). \]
    Since $f_{\min}$ vanishes on the image of $\varphi$ and $s := \ord f_{\min}$, we have $\dim(\im(\varphi)) \leqslant s$. Moreover, since $f_{\min}$ is unique up to a constant factor \cite[Proposition 1.27]{pogudin2023differential}, we have $\dim(\im(\varphi)) = s$. 
    Therefore, the image of $\varphi$ is a hypersurface and the following holds:
    \[(y - f, y' - \mathcal{L}_{\bg}(f), \ldots, y^{(s)} - \mathcal{L}_{\bg}^s(f)) \cap \mathbb{K}[\bmu, y^{(\infty)}] = (f_{\min}).\]
    
\end{proof}

\begin{lemma} \label{lem::monom-bound-param}
  For every monomial $\bmu^{\balpha} \bx^{\bbeta}$ and every monomial $ \bmu^{\tilde{\balpha}} \bx^{\tilde{\bbeta}} $ which occurs in 
  $\mathcal{L}_{\bg}(\bmu^{\balpha}\bx^{\bbeta})$
  the following inequalities hold:
  \[ |\tilde{\balpha}| \leqslant |\balpha| + D_{\bmu} \quad \text{ and } \quad |\tilde{\bbeta}| \leqslant |\bbeta| + D_{\bx} - 1.\]
\end{lemma}

\begin{proof}
 Since $\mathcal{L}_{\bg}(\bmu^{\balpha} \bx^{\bbeta}) = \bmu^{\balpha} \sum\limits_{i=1}^n g_i \frac{\partial}{\partial x_i} \bx^{\bbeta}$, 
 then $\bmu^{\tilde{\balpha}} \bx^{\tilde{\bbeta}} = \bmu^{\balpha} m \frac{\partial}{\partial x_i} \bx^{\bbeta}$ 
 for some $1 \leqslant i \leqslant n$ 
 and some monomial $m$ in $g_i$. So, $|\tilde{\balpha}| \leqslant |\balpha| + D_{\bmu}$ and $|\tilde{\bbeta}| \leqslant |\bbeta| + D_{\bx} - 1$.
\end{proof}

\begin{corollary} \label{cor::monom-bound-param}
 For every monomial $\bmu^{\balpha} \bx^{\bbeta}$ in $\mathcal{L}_{\bg}^{i}(f)$, the following inequalities hold:
 \[ |\balpha| \leqslant d_{\bmu} + i D_{\bmu} \quad \text{ and } \quad |\bbeta| \leqslant d_{\bx} + i(D_{\bx} - 1).\]
\end{corollary}
\begin{proof}
 The corollary follows by an $i$-fold application of Lemma \ref{lem::monom-bound-param} to the polynomial $f$ of degree $d_{\bmu}$ in variables $\bmu$
 and degree $d_{\bx}$ in variables $\bx$.
\end{proof}

\section{Auxiliary facts from algebraic geometry}
This section collects several algebraic geometry lemmas which will be used for the proof of Theorem \ref{th::bound-param} in Section \ref{sec:proof2}.

Throughout the section $\mathbb{A}^n$ stands for the $n$-dimensional affine space over the fixed algebraically closed field.

\begin{lemma} \label{lem:dominate}
  Let $m,n$ and $k$ be positive integers such that $n = m+k$. Let
  $X \subset \mathbb{A}^{n} = \mathbb{A}^m \times \mathbb{A}^k$ be an
 equidimensional variety of dimension $D$ and let
  $\pi : \mathbb{A}^{n} \rightarrow \mathbb{A}^{k}$ be the projection onto the
  last $k$ coordinates. 
  Denote $Y := \overline{\pi(X)}$ and suppose that $Y$
  is equidimensional of dimension $d \leqslant D$. 
  Then, for a generic affine space
  $L$ in $\mathbb{A}^m$ of codimension $D - d$, we have that $X \cap L$
  projects dominantly to $Y$ and $\dim (X \cap L) = \dim Y$.
\end{lemma}

\begin{proof}
  We work by induction over the quantity $D-d$, if $D = d$ there is nothing to prove. 
  
  Choose an irreducible component $Z$ of $Y$ and an irreducible
  component $W$ of $X$ projecting dominantly to $Z$. 
  By \cite[Theorem
  1.25(ii)]{shafarevich1994basic} there is
  an open subset $U$ of $Z$ such that for all $\by^{\ast}\in U$ we have
  $\dim(\pi^{-1}(\by^{\ast})\cap W) = D - d$. 
  Viewing
  $\pi^{-1}(\by^{\ast}) \cap W$ as a variety in $\mathbb{A}^m$, for a
  hyperplane $H$ in $\mathbb{A}^m$ chosen from an open subset $\widetilde{U}$ of
  all hyperplanes in $\mathbb{A}^m$ we have that $H$ intersects
  $\pi^{-1}(\by^{\ast}) \cap W$ transversally so that
  $\dim(\pi^{-1}(\by^{\ast}) \cap W\cap H) = D - d - 1$. In particular there is
  $\bx^{\ast}\in W$ such that $H(\bx^{\ast})\neq 0$. Since $W$ is irreducible
  this implies that $W\cap H$ is equidimensional of dimension $D-1$. This implies, using the theorem of
  dimension of base and fiber \cite[Theorem 10.10]{eisenbudCommutativeAlgebra1995}, that we have
  \[
  \dim(\pi(W\cap H)) \geqslant \dim(W\cap H) - \dim(\pi^{-1}(\by^{\ast}) \cap W\cap H) = (D - 1) - (D - d - 1) = d.
  \]
  Hence, by the irreducibility of $Z$, $\pi(W\cap H)$ is
  dense in $Z$. 
  Potentially shrinking the open subset $\widetilde{U}$ out of which
  $H$ is chosen, these assertions hold simultaneously for all
  irreducible components $W$ of $X$ projecting dominantly to $Z$ with
  the same choice of $H\in \widetilde{U}$. 
  Potentially shrinking $\widetilde{U}$ further, these
  assertions hold simultaneously for all irreducible components $W$ of $X$ that project dominantly to some component of $Y$ with the same choice of $H\in \widetilde{U}$.

  In order to ensure the dimensionality drop on the whole $X$, we
  choose a point $\bx^{\ast}_{W}\in W$ for each $W$ an irreducible component of $X$ that does not project
  dominantly to some irreducible component of $Y$. 
  After potentially shrinking
  $\widetilde{U}$ we have $H(\bx^{\ast}_{W})\neq 0$ for all $H\in \widetilde{U}$ and all such $W$. 
  By construction, for each component $W$ of $X$ either $H$ intersects $W$ transversally or $W\cap H= \emptyset$. 
  This implies that $X\cap H$ is equidimensional of dimension $D - 1$
  and projects dominantly to $Y$. 
  We now conclude using the induction
  hypothesis, replacing $X$ with $X\cap H$.
\end{proof}

\begin{lemma}
\label{lem:fiberdim}
  Let $m,n$ and $k$ be positive integers such that $n = m+k$. Let
  $X \subset \mathbb{A}^{n} = \mathbb{A}^m \times \mathbb{A}^k$ be a
 variety of dimension $d$ and let
  $\pi : \mathbb{A}^{n} \rightarrow \mathbb{A}^{k}$ be the projection onto the
  last $k$ coordinates. 
  Denote $Y := \overline{\pi(X)}$ and suppose that $Y$
  is equidimensional of dimension $d$. 
  Then there is a Zariski open dense subset $U$ of $Y$
  such that for all $\by^{\ast}\in U$ we have that the fiber $\pi^{-1}(\by^{\ast})\cap X $ is finite and non-empty.
\end{lemma}

\begin{proof}
Without loss of
  generality, we may replace $X$ by the union $X'$ of all irreducible
  components of $X$ that project dominantly to some irreducible
  component of $Y$. Indeed, for a general $\by^{\ast}\in Y$ we have
  $\pi^{-1}(\by^{\ast})\cap X = \pi^{-1}(\by^{\ast})\cap X'$.

  Let $Z$ be an irreducible component of $Y$. Each irreducible
  component $W$ of $X$ projecting dominantly to $Z$ necessarily has dimension $d$, and so by \cite[Theorem
  1.25(ii)]{shafarevich1994basic}, there is an open subset $U_W$ of
  $Z$ such that for every $\by^{\ast}\in U_W$ we have
  $\dim(\pi^{-1}(\by^{\ast}) \cap W) = 0$. After possibly shrinking $U_W$ the fiber $\pi^{-1}(\by^{\ast}) \cap W$ is in addition non-empty by \cite[Lemma 4.3(i)]{hong2020global}.
  Let $U_Z$ be the intersection
  of all such $U_W$. We then have
  $\dim(\pi^{-1}(\by^{\ast}) \cap X) = 0$ for all
  $\by^{\ast}\in U_Z$. Taking the union over all $U_Z$, with $Z$ running
  over the irreducible components of $Y$, we obtain a dense open
  subset $U$ of $Y$ such that for every $ \by^{*}\in U$ we have
  $\dim(\pi^{-1}(\by^{\ast}) \cap X) = 0$ for all $\by^{\ast}\in U$ and $\pi^{-1}(\by^{\ast}) \cap X$ is non-empty.
\end{proof}

\begin{definition}
    For a polynomial $f(\mathbf{x}) \in \mathbb{K}[x_1, \ldots, x_m]$ of degree $d$, we define its \emph{homogenization in $\mathbf{x}$} using an additional variable $z$ as
	\begin{equation*}
			f^h(x_1, \ldots, x_m, z) := z^{d} f(\dfrac{x_1}{z}, \ldots, \dfrac{x_m}{z}). 
	\end{equation*}
\end{definition}

\begin{lemma} \label{lem::degree-new} Let $m$, $k$ be positive
  integers and $\bx = [x_1, \ldots, x_{m}]^T$,
  $\by = [y_1, \ldots, y_{k}]^T$. 
  Denote by $\pi_{\by}$ the projection map
  corresponding to the inclusion
  $\mathbb{K}[\by]\hookrightarrow \mathbb{K}[\bx,\by]$. 
  Let
  $p_1, \ldots, p_{m+1}$ be polynomials in $\mathbb{K}[\bx, \by]$ such that
  $\deg_{\bx} p_i := d_{ix} \in \mathbb{Z}_{> 0}$ and
  $\deg_{\by} p_i := d_{iy} \in \mathbb{Z}_{\geqslant 0}$.  Suppose in addition that the ideal 
  $I:=(p_1,\dots,p_{m+1})$ has dimension $k-1$, suppose that
  $I \cap \mathbb{K}[\by] = (g)$ for some square-free $g\in \mathbb{K}[\by]$.
  
  Then,
  \begin{equation} \label{eq::mult-bound-new}
    \deg g \leqslant \sum_{i=1}^{m + 1} d_{iy} \prod\limits_{j=1, j \neq i}^{m+1} d_{ix}. 
  \end{equation}
\end{lemma}

\begin{proof}
Let $X = \mathbb{V}(I)$. 
Denote by $\pi_{\by}\colon \mathbb{A}^{m + k} \to \mathbb{A}^k$ the projection map
  corresponding to the inclusion
  $\mathbb{K}[\by]\hookrightarrow \mathbb{K}[\bx,\by]$. 
We will find an affine line $L$ in $\mathbb{A}^k$ such that
\begin{align*}
  (1)& \quad \forall \by^\ast \in L \cap \mathbb{V}(g)\colon \pi_{\by}^{-1}(\by^{\ast}) \cap X \text{ is finite and nonempty},\\
  (2) & \quad \# (L \cap \mathbb{V}(g)) = \deg g.
\end{align*}
We note that $\overline{\pi(X)} = \mathbb{V}(I \cap \mathbb{K}[\by]) = \mathbb{V}(g)$, 
so $\pi(X)$ is equidimensional of dimension $k - 1$. By Lemma \ref{lem:fiberdim} there exists a Zariski open dense subset $U$
of $\pi(X)$ such that for all $\by^{\ast} \in U$ the fiber $\pi^{-1}(\by^{\ast}) \cap \mathbb{V}(I)$
is finite and non-empty.

Since $g$ is a nonzero square-free polynomial, by \cite[ Chapter 9, \S 4, Exercise 12]{cox1997ideals} the second property is true for a generic $L$ (that is, it holds on a nonempty open subset in the set of lines).
After possibly shrinking this open subset of lines, $L$ intersects $\mathbb{V}(g)$ only in $U$ and so the first property holds as well. 

Fix one such $L$.
Due to the finiteness and nonemptyness of the fibers, the set $L \cap \mathbb{V}(I)$ is finite and its cardinality is at least $\#(L \cap \mathbb{V}(g)) = \deg g$. 
The cardinality of $L \cap \mathbb{V}(I)$ is bounded by the number
of isolated solutions of the bihomogeneous
systems of equations consisting of $p_1^h=0,\dots,p_{m+1}^h=0$, where $\bullet^h$ denotes bihomogenization with respect to the variables $\bx$ and $\by$, and the homogenized (only in $\by$) defining equations of $L$. 
This number is in turn bounded by the multi-homogeneous B\'ezout bound associated to this system \cite[p. 106]{morgan1987homotopy}. 
Thus, we have  
\[
  \deg g = \#(L \cap \mathbb{V}(g)) \leqslant \coeff_{v^{m} u^{k}} P(u, v),
\]
 where $P(v, u) = u^{k-1}\prod\limits_{i=1}^{m+1} (d_{ix} v + d_{iy} u )$.
To find the coefficient that corresponds to the monomial $v^{m} u^{k}$ we choose exactly $m$ of $m + 1$ factors in the product of linear forms
to contribute the term $d_{ix} v$ and the remaining one to contribute the term $d_{iy} u$. 
This gives the desired
\[ 
\deg g = \#(\mathbb{V}(g) \cap L)\leqslant \sum_{i=1}^{m + 1} d_{iy} \prod\limits_{j=1, j \neq i}^{m+1} d_{ix}.
\]
\end{proof}

\section{Proof of the bound}\label{sec:proof2}

\begin{proof}[Theorem \ref{th::bound-param}]
  Let us denote by $\nu$ the order of the minimal polynomial $f_{\min}$ and recall that $\nu \leqslant n$ by \cite[Theorem 3.16 and Corollary 3.21]{hong2020global}.
  By Lemma \ref{lem::ideal-equivalence}, we get 
  \[ 
  J := (y - f, y' - \mathcal{L}_{\bg}(f) \ldots, y^{(\nu)} - \mathcal{L}_{\bg}^{\nu}(f)) \cap \mathbb{K}[\bmu, y^{(\leqslant \nu)}] = I_{\bg, f} \cap \mathbb{K}[\bmu, y^{(\leqslant \nu)}] 
  \]
  and $I_{\bg, f} \cap \mathbb{K}[\bmu, y^{(\leqslant \nu)}] = (f_{\min})$. Thus, the ideal $J = (f_{\min})$
  is prime and principal.
 
 Consider some $[\omega_1, \dots, \omega_{\nu}]^T \in \mathbb{Z}_{\geqslant 0}^{\nu}$ (to be specified later) and  define a $\mathbb{K}$-algebra homomorphism $\varphi : \mathbb{K}[\bmu,\bx,  y^{(\leqslant \nu)}] \mapsto \mathbb{K}[\bmu,\bx, z_0, z_1, \ldots, z_{\nu}]$, such that
  \begin{align*}
    \mu_i & \mapsto \mu_i, \\
    x_i & \mapsto x_i, \\
    y^{(i)} & \mapsto p_i(z_i), \; p_i \in \mathbb{K}[z_i] \text{ and }\deg p_i(z_i) = \omega_i.
  \end{align*}
  According to \cite[Lemma 4]{mukhina2025projecting} we can choose $p_i$ such that $\tilde{f}_{\min} := \varphi(f_{\min})$ is square-free.
  We define $\tilde{f}_i := \varphi(y^{(i)} - \mathcal{L}_{\bg}^{i}(f))$ for $0 \leqslant i \leqslant \nu$. 
  By Corollary \ref{cor::monom-bound-param} for every $0 \leqslant i \leqslant \nu$ we have  $\deg_{ \bmu} \mathcal{L}_{\bg}^{i}(f) \leqslant d_{\bmu} + i D_{\bmu}$ and $\deg_{\bx} \mathcal{L}_{\bg}^{i}(f) \leqslant d_{\bx} + i(D_{\bx} - 1).$
  Thus, 
  \[ \deg_{\by} \tilde{f}_k \leqslant \omega_k,\;\; \deg_{\bmu} \tilde{f}_k  \leqslant d_{\bmu} + k D_{\bmu}\; \text{ and } \;\deg_{ \bx} \tilde{f}_k  \leqslant  d_{\bx} + (k -1) D_{\bx}.\]

The rest of the proof will be divided into three cases corresponding to different inequalities among \eqref{eq::bound_multi}-\eqref{eq::bezout_y_only}.     
    \paragraph{Case 1.} Let $\omega_i = d_{\bmu} + i D_{\bmu}$ for $i = 1, \ldots, \nu$. 
    Then, $\deg_{\bmu, \by} \tilde{f}_k \leqslant d_{\bmu} + k D_{\bmu}$.

    Consider the ideal $\tilde{I} := (\tilde{f}_0, \ldots, \tilde{f}_{\nu})$.
    We fix the monomial ordering on $\mathbb{K}[\bmu, \bx, \bz]$ to be the lexicographic monomial ordering with 
      \[ z_{\nu} > \ldots > z_1 > z_0 > x_n > \ldots > x_1 > \mu_1 > \ldots > \mu_{r}.\]
    Then $\initial(\tilde{f}_i) = c_i z_i^{\omega_i}$ for some $c_i \in \mathbb{K}$ and
    $\initial(\tilde{f}_0), \ldots, \initial(\tilde{f}_{\nu})$ form a regular sequence. 
    Thus, by \cite[Pr. 15.15]{eisenbudCommutativeAlgebra1995} $\tilde{f}_0, \ldots, \tilde{f}_{\nu}$
    form a regular sequence as well.  
    Since $\tilde{f}_0, \ldots, \tilde{f}_{\nu}$ is a regular sequence, 
     $\dim \tilde{I} = 2 n + r - \nu$ and the ideal $\tilde{I}$ is equidimensional.
    We denote by $\pi : \mathbb{A}^{2n + r + 1} \rightarrow \mathbb{A}^{n+r+1}$ the projection onto $[\bmu, \bz]^T$
    coordinates and let $Y := \pi(\mathbb{V}(\tilde{I})) = \mathbb{V}(\tilde{I} \cap \mathbb{K}[\bmu, \bz]) =  \mathbb{V}(\tilde{f}_{\min})$.
    Then $Y$ is equidimensional of dimension $n + r$.
    By Lemma \ref{lem:dominate} for a generic affine space $L$ in $\mathbb{A}^{n}$ of codimension $n - \nu$, the projection of
    $\mathbb{V}(\tilde{I}) \cap L$ to $Y$ is dominant and $\dim (\mathbb{V}(\tilde{I}) \cap L) = \dim Y = n + r$. 
    We can rewrite $L$ as $\mathbb{V}(h_1, \ldots, h_{n-\nu})$ for some polynomials $h_i$ in $\mathbb{K}[\bx]$ of degree one.
    Thus, the ideal $\tilde{I} + I(L) = (\tilde{f}_0, \ldots, \tilde{f}_{\nu}, h_1, \ldots, h_{n-\nu})$
    has dimension $n + r$.
    Applying  Lemma \ref{lem::degree-new}
     with $p_i = \tilde{f}_i$ for $0 \leqslant i \leqslant \nu$, $p_{\nu + j} = h_j$ for $1 \leqslant j \leqslant n - \nu$
    with the two sets of variables $G_1 = [\bx]$ and $G_2 = [\bmu, \bz]^T$, we obtain
  \[
  \deg \tilde{f}_{\min} \leqslant
  \sum\limits_{i=0}^{n + 1}d_{iy} \prod\limits_{j = 0, j \neq i}^{n + 1}d_{ix},
  \]
  here $d_{iy} = \deg_{\bmu, \bz} p_i$ and $d_{ix} = \deg_{\bx} p_i$. 
  We note that
  \[
    d_{iy} = \begin{cases}
        d_{\bmu} + i D_{\bmu}, \text{ for } 0 \leqslant i \leqslant \nu,\\
        0, \text{ for } i > \nu
    \end{cases}\quad\text{ and }\quad
    d_{ix} = \begin{cases}
        d_{x} + i (D_{x} - 1), \text{ for } 0 \leqslant i \leqslant \nu,\\
        1, \text{ for } i > \nu
    \end{cases}
  \]
  Hence, we have
  \[
  \deg \tilde{f}_{\min} \leqslant
  \sum\limits_{i=0}^{\nu}(d_{\bmu} + i D_{\bmu}) \prod\limits_{j = 0, j \neq i}^{\nu}(d_{\bx} + j(D_{\bx}-1)).
  \]
  By applying the homomorphism $\varphi$ to a monomial $m = \bmu^{\bm{\ell}} y^{e_0} (y')^{e_1} \ldots (y^{(\nu)})^{(e_{\nu})}$ in the support 
  of $f_{\min}$ we get $\varphi(m) = c \bmu^{\bm{\ell}}z_0^{\omega_0 e_0} z_1^{\omega_1 e_1} \ldots z_{\nu}^{\omega_{\nu} e_{\nu}} + q(\bmu, \bz)$
  with $c \in \mathbb{K}^\ast$ and $\deg q < \sum_{i = 1}^{r} \ell_i + \sum_{i=0}^{\nu} \omega_i e_i$. 
  Using the established degree bound for
  $\tilde{f}_{\min}$, we obtain \eqref{eq::bound_multi}:
  \[ \sum_{i = 1}^{r} \ell_i  + \sum_{i=0}^{\nu} \omega_i e_i = \sum_{i = 1}^{r} \ell_i  + \sum_{i=0}^{\nu} (d_{\bmu} + i D_{\bmu}) e_i \leqslant 
  \sum\limits_{i=0}^{\nu}(d_{\bmu} + i D_{\bmu}) 
      \prod\limits_{j = 0, j \neq i}^{\nu}(d_{\bx} + j(D_{\bx} - 1)). \]

     \paragraph{Case 2.} Now we consider $\omega_i = d_{\bx} + i (D_{\bx} - 1)$. In this case $\deg_{\bmu} \tilde{f}_k \leqslant d_{\bmu} + k D_{\bmu}$ and $\deg_{\bx, \bz} \tilde{f}_k \leqslant d_{\bx} + k (D_{\bx} - 1)$ and, thus, the total degree 
     \[
    \deg \tilde{f}_k \leqslant d_{\bx} + d_{\bmu} + k(D_{\bx} + D_{\bmu} - 1).
    \]
    Using \cite[Lemma 5]{mukhina2025projecting} with $p_i = \tilde{f}_i$
  we get $\deg \tilde{f}_{\min} \leqslant \prod_{i=0}^{\nu} \deg \tilde{f}_i$.

  By applying the homomorphism $\varphi$ to a monomial $m = \bmu^{\bm{\ell}} y^{e_0} (y')^{e_1} \ldots (y^{(\nu)})^{(e_{\nu})}$ in the support 
  of $f_{\min}$ we get $\varphi(m) = c \bmu^{\bm{\ell}} z_0^{\omega_0 e_0} z_1^{\omega_1 e_1} \ldots z_{\nu}^{\omega_{\nu} e_{\nu}} + q(\bmu, \bz)$
  with $c \in \mathbb{K}^\ast$ and $\deg q < \sum_{i = 1}^{r} \ell_i + \sum_{i=0}^{\nu} \omega_i e_i$.
  Using the established degree bound for
  $\tilde{f}_{\min}$, we obtain \eqref{eq::bound_bezout_global}:
  \[ \sum_{i = 1}^{r} \ell_i +\sum_{i=0}^{\nu} \omega_i e_i = \sum_{i = 1}^{r} \ell_i + \sum_{i=0}^{\nu} (d_{\bx} + i(D_{\bx}-1)) e_i \leqslant \prod_{i=0}^{\nu} (d_{\bx} + d_{\bmu} + i(D_{\bx} + D_{\bmu} - 1)). \]

  \paragraph{Case 3.} 
  Let us consider the system over the field $\mathbb{K}(\bmu)$, that is, the parameters will be a part of the coefficient field.
  Then we can apply the result of the previous case having $r = 0$, $d_{\bmu} = D_{\bmu} = 0$.
  This gives us precisely \eqref{eq::bezout_y_only}.
\end{proof}

\section{Algorithm}\label{sec:algorithm}

In this section, we show how the bounds from Theorem \ref{th::bound-param} and Corollary \ref{cor::bound} can be used to compute the minimal differential equation for $y = f(\bmu, \bx)$ for a system of differential equations of the form
\begin{equation} \label{eq::ODE}
    \bx' = \bg(\bmu, \bx),
\end{equation}
where $\bx = [x_1, \ldots, x_n]^T$, $\bmu = [\mu_1, \ldots, \mu_r]^T$, $\bg \in \KK[\bmu, \bx]^n$, and $f \in \KK[\bmu, \bx]$. 

Our algorithm \cite[Algorithm 1, 2]{mukhina2025projecting} together with the proofs of the termination and correctness extends straightforwardly to the more general case we consider in this paper once a bound on the support of the output has been established.
Therefore, instead of repeating the formal description of the algorithm and technical details from \cite{mukhina2025projecting}, we will explain how it works with the new bound on a simple example.

 \begin{algorithmic}[1]
    \Require  A system
    \begin{equation*}
        \begin{cases}
            x_1'  = x_1 + 8 x_2,\\
            x_2'  = 7 x_1 + x_2,\\
            y = x_1 + x_2.
        \end{cases}
    \end{equation*}
    \Ensure the minimal polynomial $f_{\min} \in (x_1' - x_1 - 8 x_2, x_2'  - 7 x_1 - x_2)^{(\infty)} \cap\mathbb{Q}[y^{(\infty)}]$
    
    \State Compute the order of $f_{\min}$ via $\nu :=  \operatorname{rank}(\frac{\partial}{\partial x_j}\mathcal{L}_{\bg}^{i - 1}(f))_{i, j = 1}^2 = 2$ (see Notation \ref{not::LieOperator}).

    \State Compute the support $S$ of $f_{\min}$ using Corollary \ref{cor::bound} and $\nu$: $S := \{1, y, y', y'' \}$.

    \State Make an anzats $f_{\min} = \gamma_1 + \gamma_2 y + \gamma_3 y' + \gamma_4 y''$.

    \State Choose random points $p_1, p_2, p_3, p_4 \in \mathbb{Q}^2$.
        
    \State Express $y, y'$ and $y''$ via $x_1, x_2$: 
    \[y = x_1 + x_2, \quad y' = x_1' + x_2' = 8 x_1 + 9 x_2, \quad y'' = 8 x_1' + 9 x_2' = 71 x_1 + 73 x_2.\]

    \State Evaluate $y(p_i), y'(p_i), y''(p_i)$ for each $1 \leqslant i \leqslant 4$ and plug in $f_{\min}$.

    \State Solve the resulting linear system
    \[ \begin{pmatrix}
        1& 3 & 26 & 217\\
        1& 11 & 95 & 795\\
        1& 16 & 139 & 1158\\
        1& 12 & 105 & 870\\
        \end{pmatrix}
        \begin{pmatrix}
            \gamma_1\\
            \gamma_2\\
            \gamma_3\\
            \gamma_4
        \end{pmatrix} = 
        \begin{pmatrix}
            0\\
            0\\
            0\\
            0
        \end{pmatrix}
            \]
    The solution space is spanned by $\gamma_1 = 0, \gamma_2 = - 55, \gamma_3 = -2, \gamma_4 = 1$.

   \Return $f_{\min} = -55 y - 2 y' + y''$.
 \end{algorithmic}

\begin{remark}
  If the input system has parameters (i.e. $r \neq 0$), then in Step 2 of our algorithm, we replace the bound from Corollary \ref{cor::bound} with that from Theorem \ref{th::bound-param}. 
  Moreover, if $r = 0$, and $y$ is equal to one of $\bx$, we instead use the bound given in \cite[Theorem 1]{mukhina2025projecting}.
\end{remark}

\section{Implementation and performance}\label{sec:implementation}

We used the new bounds given by Theorem \ref{th::bound-param} and Corollary \ref{cor::bound} as described in Section \ref{sec:algorithm} to extend the proof-of-concept implementation, DiffMinPoly, of a differential elimination algorithm from \cite{mukhina2025projecting} (based on Oscar \cite{OSCAR}, Nemo \cite{10.1145/3087604.3087611}, and Polymake \cite{gawrilow2000polymake} libraries) to the larger class of systems considered in this paper. 
The source code and instructions for this new version of software together with the models used in this section are publicly available at

\begin{center}
    \url{https://github.com/ymukhina/Loveandsupport/tree/y-input}
\end{center}

The goal of the present section is to show that this implementation can perform differential elimination in reasonable time on commodity hardware for some instances which are out of reach for the existing state-of-the-art software thus pushing the limits of what can be computed.
Note that we are not aiming here at comprehensive benchmarking of differential elimination algorithms and we have deliberately chosen benchmarks allowing us to highlight the advantages of the present method. We discuss the limitations of our approach at the end of the section.

We will use four sets of models:
\begin{itemize}
    \item \emph{Dense models}.
    For fixed $n, D, d$ we define $\operatorname{Dense}_n(D, d)$
    to be a system of the form $\bx' = \bg(\bx),\;  y = f(\bx)$, where the dimension of $\bx$ is $n$, $f$ is a random dense polynomial of degree $d$ and $g_1, \ldots, g_n$ are random dense polynomials of degree $D$, where the coefficients are sampled independently uniformly at random from $[-100,100]$.
    Here is, for example, an instance of $\operatorname{Dense}_2(1,2)$:
    \begin{align*}
        x_1' &=  - 29 x_1 + 43 x_2 - 5,\\
        x_2' &= 5 x_1 - 87 x_2 - 36,\\
        y &= 32 x_1^2 - 16 x_1 x_2 + 8 x_1 - 92 x_2^2 + 3 x_2 + 67.
    \end{align*}
    
    \item \emph{$\mu_0$--Dense models}.
    For fixed $n, D_{\bx}, d_{\bx}, |\bmu|,$ we define $\operatorname{\mu_0--Dense}_n(D_{\bx}, d_{\bx}, |\bmu|)$ 
    to be a system of the form $\bx' = \bg(\bmu, \bx), \; y = f(\bmu, \bx)$, where 
    \begin{itemize}
        \item the dimension of $\bx$ is $n$;
        \item $f$ is a random dense polynomial with degree $d_{bx}$ in $\bx$ and degree $0$ in $\bmu$;
        \item $g_1, \ldots, g_n$ are random dense polynomials of degree $D_{\bx}$ in $\bx$ and degree $1$ in $\bmu$;
        \item all polynomial coefficients are sampled independently uniformly at random from $[-100,100]$.
    \end{itemize}
   For example, an instance of $\operatorname{\mu_0--Dense}_2(1,2,1)$:
    \begin{align*}
        x_1' &= 37 a_1 x_1 - 9 x_1 - 28 a_1 x_2 + 52 x_2 + 73 a_1 - 46,\\
        x_2' &= 69 x_1 a_1 - 43 x_1 - 36 a_1 x_2 + 91 x_2 + 79 a_1 - 69,\\
        y &= 31 x_1^2 - 34 x_1 x_2 + 60 x_1 - 74 x_2^2 + 96 x_2 + 58.
    \end{align*}

    \item \emph{$\mu_1$--Dense models} are defined in the same way as $\mu_0$--Dense models with the only difference that $f$ is of degree 1 in $\bmu$.

    \item \emph{Competing species with nonlinear observations.} 
    We will start with the following parametric model used, for example, to model populational dynamics of competing species:
    \begin{equation}\label{eq:LV}
    \begin{cases}
        x_1' = x_1(a_1 + a_2 x_1 + a_3 x_2),\\
        x_2' = x_2(b_1 + b_2 x_1 + b_3 x_2).
    \end{cases}
    \end{equation}
    We will consider the following test cases involving this model:
    \begin{itemize}
        \item[CS1:] We will take all the parameters except for $a_3$ and $b_2$ (inter-species interaction rates) to be random numbers from $\{\frac{1}{10}, \ldots, \frac{10}{10}\}$ and set $a_3 = -b_2 = a$ to be an unknown parameter.
        The minimal differential equation will be computed for the nonlinear observation function $y = a_2 x_1^2 + b_3 x_2^2$ corresponding to the intra-species interactions.

        \item[CS2-3:] We use the model \eqref{eq:LV} now with all the parameters to be fixed to random scalars from $\{\frac{1}{10}, \ldots, \frac{10}{10}\}$.
        We will compute minimal differential equation for nonlinear observations following the power law $y = x_1^3 + x_2^2$ for CS2 and $y = x_1^4 + x_2^4$ for CS3.
        We do not claim any biological interpretation for these observation functions, they are used as examples of sparse nonlinear expressions.
    \end{itemize}
\end{itemize}

The specific randomly generated instances used for the experiments can be found in the repository.
For comparison, we used the following software packages allowing to perform (among other things) differential elimination: DifferentialThomas \cite{bachler2012algorithmic} (part of Maple, we used Maple 2023), DifferentialAlgebra \cite{diffalg} (written in C++ and Python, we used version 4.1), and StructuralIdentifiability \cite{Dong2023} (written in Julia, we used version 0.5.12). All computations were performed on a single core of an Apple M2 Pro processor with 32 GB of memory.

Tables \ref{tab:non-param}, \ref{tab:case-1}, and \ref{tab:case-2} report the performance of the selected software tools for computing the minimal polynomial in the non-parametric case, in the case with $d_{\bmu} = 0$, $D_{\bmu} = 1$, and in the case with $d_{\bmu} = D_{\bmu} = 1$, respectively.

\begin{table}[htbp!]
\centering
	\begin{tabular}{ l| c c c c } 
	\hline
    Name & SI.jl & Diff.Thomas & Diff.Algebra & DiffMinPoly (our) \\ 
	\hline
    
    $\operatorname{Dense}_2(4,2)$ & 414 & $>$ 2h & RE & 11\\
    
    $\operatorname{Dense}_2(2,3)$ & OOM & $>$ 2h  & RE & 3\\
    
    $\operatorname{Dense}_2(3,3)$ & OOM & $>$ 2h & OOM & 44\\

    $\operatorname{Dense}_2(2,4)$ & OOM & $>$ 2h & OOM & 92\\

    $\operatorname{Dense}_2(3,4)$ & OOM & $>$ 2h & OOM & 991 \\

    CS2 & 2198 & $>$ 2h & OOM & 0.8 \\

    CS3 & OOM & $>$ 2h & OOM &  14\\
    
    \hline
    
\end{tabular}
\caption{Runtimes in non-parametric case (in seconds if not written explicitly)\\ \emph{OOM = ``out of memory'', RE = ``runtime error''}}\label{tab:non-param}
\end{table}

\begin{table}[htbp!]
\centering
	\begin{tabular}{ l| c c c c } 
	\hline
    Name & SI.jl & Diff.Thomas & Diff.Algebra & DiffMinPoly (our) \\ 
	\hline
    $\operatorname{\mu_0--Dense}_2(1,3,1)$ & OOM & $>$ 2h & OOM & 69\\

    $\operatorname{\mu_0--Dense}_2(2,2,1)$ & 1332 & $>$ 2h & OOM & 21\\

    $\operatorname{\mu_0--Dense}_2(3,2,1)$ & OOM & $>$ 2h & OOM & 1208\\

     $\operatorname{\mu_0--Dense}_2(2,2,2)$ & OOM &  $>$ 2h & OOM & 2426\\

    CS1 & 128 &  $>$ 2h & OOM & 3\\
    \hline  
\end{tabular}
\caption{Runtimes for $d_{\bmu} = 0, D_{\bmu} = 1$ (in seconds if not written explicitly)\\ \emph{OOM = ``out of memory''}}\label{tab:case-1}
\end{table}

Tables \ref{tab:non-param}-\ref{tab:case-2}  show that our algorithm can significantly outperform the state-of-the-art methods on appropriate benchmarks.
On the other hand, we must mention the following two important limitations of the current version of our algorithm:
\begin{itemize}
    \item The systems used for benchmarking in this section are dense or moderately sparse. 
    If the level of sparsity is more substantial, in particular, if the degrees of the polynomials in $\bg$ vary, the bound becomes too conservative, and other methods perform better.
    One way to mitigate this issue is to take into account more detailed information on the supports of $\bg$ and $f$, promising preliminary results in this direction are reported in \cite[Section 5]{mohr2025computationnewtonpolytopeseliminants}. 

    \item Similarly, if the number of parameters increases (as in applications to structural identifiability), the bound becomes too conservative as well.
    In other words, the current approach does not take into account the sparsity with respect to the parameters.
    One possible workaround is, again, to refine a bound taking the sparsity into account.
    An alternative is to use sparse polynomial interpolation to reconstruct the coefficients of the eliminant with respect to $y^{(\infty)}$ by evaluating the parameters at appropriate linear forms of a single parameter.
\end{itemize}

\begin{table}[htbp!]
\centering
	\begin{tabular}{ l| c c c c } 
	\hline
    Name & SI.jl & Diff.Thomas & Diff.Algebra & DiffMinPoly (our) \\ 
	\hline
    $\operatorname{\mu_1--Dense}_2(1,3,1)$ & OOM &  $>$ 2h & OOM & 351 \\
    $\operatorname{\mu_1--Dense}_2(2,2,1)$ & 5147 & $>$ 2h  & OOM & 107 \\
    $\operatorname{\mu_1--Dense}_2(1,2,2)$ & 764 &  $>$ 2h & OOM & 54 \\

    \hline  
\end{tabular}
\caption{Runtimes for  $d_{\bmu} = D_{\bmu} = 1$ (in seconds if not written explicitly)\\ \emph{OOM = ``out of memory''}}\label{tab:case-2}
\end{table} 

\section{Conclusion}

We present an evaluation-interpolation approach to computing the minimal differential equation for an important case of the differential elimination problem.
Namely, for polynomial parametric dynamical systems with polynomial observations.
We do this by establishing a bound for the Newton polytope of such a minimal equation. 
Numerical data from computational experiments show that the predicted number of terms is often very close to the actual number.
Our approach allows to efficiently perform elimination for realistic systems by avoiding expression swell often jeopardizing the performance of the state of the art methods.
We provide a publicly available implementation of our algorithm.

In the parametric case, while the bound is relatively accurate for one or two parameters, it becomes too conservative in the realistic scenario with multiple parameters.
One workaround would be to reduce to the case of fewer parameters by additional evaluation-interpolation on the level of coefficients. 
We leave this question for future research.
Another trait of models appearing in the modeling literature is their sparsity and, in particular, the fact that often not all the equations have the same degree.
Thus, another important challenge would be to refine the bound to take this information into account.

\bibliographystyle{splncs04}
\bibliography{mybibliography}

\end{document}